\documentclass[11pt]{article}

\usepackage{amsmath} 
\usepackage{amsthm}

\usepackage{amsthm, amssymb, amsmath, amsfonts, url, enumerate, mathtools, graphics, xcolor}
\usepackage[margin=0.75in]{geometry}
\usepackage{verbatim}
\usepackage{fullpage}
\usepackage{appendix}

\newtheorem{theorem}{Theorem}[section]

\newtheorem{corollary}[theorem]{Corollary}
\newtheorem{lemma}[theorem]{Lemma}

\theoremstyle{definition}
\newtheorem{definition}[theorem]{Definition}

\newenvironment{prf}{\noindent{\bf Proof.~}}{\(\qed\)}
\newcommand{\BPF}{\begin{prf}} 
\newcommand {\EPF}{\end{prf}}

\def\bf{{\mathbf f}}

\usepackage[ruled,vlined,linesnumbered]{algorithm2e}

\usepackage{enumitem}

\title{Sublinear-Time Approximation for Graph Frequency Vectors in Hyperfinite Graphs}

\author{
Gregory Moroie, University of Waterloo, gmmoroie@uwaterloo.ca
}

\date{August 7, 2025}

\begin{document}

\maketitle

\vspace{-15pt}

\begin{abstract}
In this work, we address the problem of approximating the $k$-disc distribution (``frequency vector") of a bounded-degree graph in sublinear-time under the assumption of hyperfiniteness. We revisit the partition-oracle framework of Hassidim, Kelner, Nguyen, and Onak \cite{hassidim2009local}, and provide a concise, self-contained analysis that explicitly separates the two sources of error: (i) the cut error, controlled by hyperfiniteness parameter $\phi$, which incurs at most $\varepsilon/2$ in $\ell_1$-distance by removing at most $\phi |V|$ edges; and (ii) the sampling error, controlled by the accuracy parameter $\varepsilon$, bounded by $\varepsilon/2$ via $N=\Theta(\varepsilon^{-2})$ random vertex queries and a Chernoff and union bound argument. Combining these yields an overall $\ell_1$-error of $\varepsilon$ with high probability. Algorithmically, we show that by sampling $N=\lceil C\varepsilon^{-2} \rceil$ vertices and querying the local partition oracle, one can in time $poly(d,k,\varepsilon^{-1})$ construct a summary graph $H$ of size $|H|=poly(d^k,1/\varepsilon)$ whose $k$-disc frequency vector approximates that of the original graph within $\varepsilon$ in $\ell_1$-distance. Our approach clarifies the dependence of both runtime and summary-size on the parameter $d$,$k$, and $\varepsilon$.
\end{abstract}

\section{Introduction}\label{sec:intro}

In many modern applications—from network mining to distributed monitoring—we face graphs so large that even linear-time algorithms become impractical. Property testing offers a way forward: by probing a small, random part of the input and tolerating a small error, one can decide global properties of bounded-degree graphs in time independent of the number of vertices (Goldreich \& Ron 1997) \cite{goldreich1997property}. Beyond simple yes/no tests, researchers have asked whether one can approximate rich statistics of a graph—such as the distribution of its ``small" neighbourhoods—also in sublinear time. 

The graph frequency vector problem formalizes this goal. Fix a radius $k$ and let $T$ be the number of isomorphism types of rooted $k$-discs in degree-$d$ bounded graphs. Every graph $G$ induces a normalized frequency vector $f_G\in[0,1]^T$, where $f_G(i)$ is the fraction of vertices whose radius-$k$ neighbourhood has type $i$. N. Alon (proof recited by Lov\'{a}sz \cite{lovasz2012large}) showed that for any error parameter $\varepsilon>0$, there is a constant-size graph $H$ (independent of $|G|$) whose $k$-disc vector approximated $f_G$ to within $\ell_1$-distance $\varepsilon$. However, the first truly algorithmic solution to this problem came from the local partition-oracle, $\mathcal{O}$, framework of Hassidim, Kelner, Nguyen, and Onak \cite{hassidim2009local}.

By sampling $O(\varepsilon^{-2})$ random vertices and querying $\mathcal{O}$, one recovers an approximate frequency vector in time $poly(d,k,1/\varepsilon^{-1})$. All subsequent improvements on constant-time estimation of graph frequencies have built on this partition-oracle approach—no fundamentally different sublinear framework has been found. 

In this paper we revisit the sublinear algorithm of \cite{hassidim2009local} under this light. Our main contribution is a concise, self-contained account showing that, for any fixed $d,k$ and desired accuracy $\varepsilon$, one can in time $poly(d,k,\varepsilon^{-1})$, with $N=\Theta(\varepsilon^{-2})$ random samples, construct a graph $H$ of size 
\begin{equation}
    |H|=poly(d^k,\varepsilon^{-1})
\end{equation}
whose frequency vector approximates $G$ to within $\varepsilon$. By making the dependence on hyperfinite parameters $\phi$ and $\rho(\phi)$ fully explicit we clarify the parameter dependence of the algorithm. 

\subsection{Main result and proof overview}
Throughout this paper, we will work towards and prove the following corollary:

\begin{corollary}
\cite{hassidim2009local}\label{cor:main-intro}
    Let $G=(V,E)$ be any $n$-vertex graph of maximum degree $d$, some fixed $k\geq 0$ and error parameter $\varepsilon\in(0,1)$. Suppose $G$ is a $\rho(\phi)$-hyperfinite graph, so that for every $\phi>0$ removing at most $\phi|V|$ edges splits $G$ into connected components of size at most 
    \begin{equation}
        \rho(\phi)=poly(d^k, 1/\phi).
    \end{equation}
    Then Algorithm \ref{alg:summary} below runs in time $poly\bigl(d,k,1/\varepsilon\bigr)$ and, with high probability, constructs a graph $H$ satisfying:
    
    1. $|H|=poly(d^k, 1/\varepsilon)$,

    2. the $\ell_1$ distance between the $k$-disc frequency vectors of $G$ and $H$ is at most $\varepsilon$.
\end{corollary}

\begin{algorithm}[H]
    \caption{Sublinear-Time Summary Graph Construction}
    \label{alg:summary}
    $N=\lceil C\,\varepsilon^{-2}\rceil$
    
    \For{$i=1,..,N$} {
    
        Pick $v_i$ uniformly at random from $V$
        
        $P_i \leftarrow \mathcal{O}(v_i)$
        
    }
    
    return $H \leftarrow \bigcup_{i=1}^N P_i$
\end{algorithm}

\paragraph{High-level ideas of the proof:}
The high level idea guiding the proof comes from these steps:
\begin{enumerate}
    \item Partition $G$ via oracle to remove at most $\phi|V|$ edges from $G$ resulting in a graph of connected components of size $\leq \rho(\phi)$.
    \item Any vertex farther than distance $k$ from all removed edges sees an identical $k$-disc in $G$ and the pruned graph $G'$. At most a $\phi B(d,k)$-fraction of vertices lie within distance $k$ of the cut, so choosing $\phi \leq \frac{\varepsilon}{2B(d,k)},$ bounds the $\ell_1$ error from this step by $\frac{\varepsilon}{2}$.
    \item Sample $N$ vertices uniformly at random from $V$, query the oracle to recover their $k$-discs in $H$, and let $\hat{f}$ be the empirical frequency vector. A Chernoff and union bound argument shows that if we set $N=\lceil C\varepsilon^{-2}\rceil$ for a sufficiently large constant $C$, then $\Pr\left[ \| \hat{f}-f_{H} \|_1 >\frac{\varepsilon}{2}\right] \leq o(1)$.
    \item By the triangle inequality our $\ell_1$ error is at most $\varepsilon$.
    \item We only collect at most $N$ oracle returned parts (of size $\leq \rho(\phi)$) and take their disjoint union. So, $|H|\leq N\rho(\phi)=poly(d,k,\varepsilon^{-1})$.
\end{enumerate}

We begin by invoking the partition oracle to remove at most $\phi|V|$ edges, producing a pruned graph $G'$ whose connected components all have size $\leq \rho(\phi)$. In the cut-error step, we bound the number of vertices whose radius-$k$ neighbourhoods are altered by these removals: each removed edge ``covers" at most $B(d,k)=2\sum_{i=0}^{k}d^i$ vertices, so altogether at most a $\phi B(d,k)$-fraction of vertices have altered $k$-discs in $G'$. Setting $\phi \leq \frac{\varepsilon}{2B(d,k)},$ ensures any vertex outside this set contributes zero to the $\ell_1$ difference, and those inside account for at most $\varepsilon/2$.

Next, in the sampling-error step, we draw $N$ vertices uniformly at random from $V$. For each sampled vertex $v$, we query $\mathcal{O}(v)$ to recover its connected component in $H$, and thus its $k$-disc. Denoting the total number of disc types $T=O(d^k)$, we apply a Chernoff bound $\Pr\left[ |\hat{p}_i-p_i| > \delta \right] \leq 2e^{-2N\delta}$ and then a union bound over $i=1,..,T$. Choosing $\delta=\varepsilon/(2T)$ and $N=\Theta(\varepsilon^{-2})$ forces $\| \hat{f}-f_{H} \|_1 \leq T\delta=\varepsilon/2$ with high probability. 

Finally, by the triangle inequality the two error sources sum to $\varepsilon$. Moreover, since each sample yields at most one part of size $\leq \rho(\phi)$ and there are $N$ samples, the final summary graph $H$ has size $|H|\leq N\rho(\phi)=poly(d,k,\varepsilon^{-1})$. This completes the proof.

\subsection{Related Work}
Property testing for bounded-degree graphs was pioneered by Goldreich \& Ron (1997) \cite{goldreich1997property}, who gave constant-time testers for connectivity, bipartiteness, and minor-freeness. Hassidim et al. (2009) \cite{hassidim2009local} introduced the partition-oracle framework, showing that any $(\phi, \rho(\phi))$-hyperfinite graph can be decomposed in time $poly(d,1/\phi)$, and used it to give the first sublinear-time algorithm for estimating that radius-$k$ neighbourhood distribution of a graph. Subsequent works—for instance Newman \& Sohler (2011) \cite{newman2011every} on testability of all hyperfinite properties, and Kumar, Seshadhri \& Stolman (2021) \cite{kumar2021random} on improved $poly(d/\phi)$-time oracles for minor-free classes—have all built on this oracle-based paradigm. 

\subsection{Organization}
In Section 2 we introduce notation and recall the partition-oracle guarantee. In Section \ref{sec:technical} we present our two-step approximation for constructing $H$.
Section \ref{subsec:GlobalPartition} gives the main helper function that aids in finding the hyperfinite-decomposition components, Section \ref{subsec:toLocal} translates our global partition algorithm into a local partition oracle and Section \ref{subsec:contructingH} concludes with the proof of Corollary \ref{cor:main-intro}. We conclude with open problems in Section \ref{sec:conclusion}. 

\section{Preliminaries}

\subsection{General Facts and Notations}

We will be studying degree-bounded, undirected graphs given in adjacency list representation, with query access to the adjacency list of each vertex. For a given number of vertices, $n$, of a graph we assume that we can sample vertices in $O(1)$ time. For a given vertex $v$ with $d$ edges, we assume we can sample each neighbour of $v$ in $O(d)$ time.

\begin{definition}
    Let $G=(V,E)$ be a degree-$d$ bounded graph. For any vertex $v$, its radius-$k$ disc is the induced subgraph of all nodes at distance $\leq k$. There are at most $T=O(d^k)$ disc-isomorphism types; we write $f_G\in \mathbb{R}^{T}$ to denote the normalized frequency vector of these types. We measure approximation by $\ell_1$/total-variation distance $\|f-g\|_1=\sum_i|f_i-g_i|$
\end{definition}

\begin{definition}[Hyperfinite graph \cite{hassidim2009local}]
    Let $G=(V,E)$ be a graph. For constants $\phi,x>0$, we say that $G$ is $(\phi,x)$-hyperfinite if it is possible to remove up to $\phi |V|$ edges of the graph such that the remaining graph has connected components of size at most $x$. Let $\rho$ be a function, $\rho : \mathbb{R^+} \rightarrow \mathbb{R^+}$. A graph is $\rho$-hyperfinite if for every $\phi>0$, $G$ is $(\phi,\rho(\phi))$-hyperfinite.
\end{definition}

\begin{definition}[Partition \cite{hassidim2009local}]
    We say $P$ is a partition of a set $S$ if it is a family of nonempty subsets of $S$ such that $\bigcup _{X\in P}X=S$, and for all $X,Y\in P$ either $X=Y$ or $X\cap Y =\emptyset$. We write $P[q]$ to denote the set in $P$ containing element $q\in S$.
\end{definition}

\begin{definition}[$(\rho, \phi)$-isolated neighbourhood \cite{hassidim2009local}]
    Let $G=(V,E)$ be an undirected graph, for a set $S\subseteq V$. Let $e_G(S)=|\{(u,v)\in E:u\in S, v\notin S\}|$, in other words $e_G(S)$ is the number of edges in $G$ with exactly one endpoint in $S$ and the other not in $S$. For some $v\in V$, we say $S$ is a $(\rho, \phi)$-isolated neighbourhood of $v\in V$ if $v\in S$, and the subgraph induced by $S$ is connected, $|S|\leq \rho$, and $e_G(S)\leq \phi|S|$. 
\end{definition}

\begin{definition}[Partition oracle\cite{hassidim2009local}]
    We say $\mathcal{O}$ is an $(\phi,\rho)$-partitioning oracle for a class of graphs, $\mathcal{C}$, if given query access to a graph $G=(V,E)$ in adjacency list representation. It provides query access to a partition $P$ of $G$. For a query about $v\in V$, $\mathcal{O}$ returns $P[v]$. A partition has the following properties:
    \begin{itemize}
        \item $P$ is a function of $G$ and random bits of the oracle. In particular, it does not depend on order of queries to $\mathcal{O}$.
        \item For all $v \in V$, $|P[v]|\leq \rho$ and $P[v]$ induces a connected graph in $G$.
        \item If $G$ belongs to $P$, then $|\{ (v,w)\in E:P[v]\neq P[w]\}|\leq \phi |V|$ with probability $>\frac{9}{10}$.
    \end{itemize}
\end{definition}

\begin{theorem}[Chernoff Bound \cite{Oliveira2025Lecture3}]\label{def:Chernoff}
    Let $X_1,...,X_n$ be independent random variables such that $X_i\in0,1$ for all $i\in[n]$. Let $X=\sum_{i=1}^nX_i$ and $\mu=\mathbb{E}[X]$. Then, for $0<\delta<1$,
    \begin{equation}
        \Pr[X\leq(1-\delta)\mu] \leq \exp(-\frac{\delta^2\mu}{2})
    \end{equation}
\end{theorem}

\section{Technical Section}\label{sec:technical}

\subsection{Global Partitioning Oracle}\label{subsec:GlobalPartition}
The first crucial step in designing the local partitioning oracle is creating a global partitioning algorithm (global oracle). Our algorithm, for each vertex $v_i\in V$ will partition $G$ into sets denoted $S_j$, $1\leq j\leq n$, such that $S_j$ is either a $(\rho,\phi)$-isolated neighbourhood of $v_i$ or $S_j=\{v_j\}$. The first algorithm we will introduce is a non-deterministic helper function, that given a vertex $v_i$ determines, with high probability, the $(\rho, \phi)$-isolated neighbourhood of $v_i$ if it exists. 

\textbf{Input -} We are given a degree-$d$ bounded hyperfinite graph $G=(V,E)$ in adjacency list representation, a vertex $v\in V$ and some $\rho,\phi > 0$. 

\textbf{Output -} If $v$ belongs to some $(\rho,\phi)$-isolated neighbourhood, denoted $S$, we will return $S$. Otherwise, we return $null$.

\begin{algorithm}
    \caption{Find Isolated Neighbourhood}
    \label{alg:helperFunction}
    $S \gets\{v\}$

    frontier $\gets$ neighbours($v$)

    \While{$|S|<\rho$ and frontier $\neq 0$}{
        pick $u$ from frontier uniformly at random
        
        $S \gets S\cup {u}$
        
        update frontier with neighbours of $u$ not already in $S$
        
        Compute $e_G(S)$
        
        \If{$e_G(S) \leq \phi\cdot|S|$}{
            return $S$        
        }
    }
    return $null$
\end{algorithm}

\begin{lemma}[High-probability recovery]
    Let $G=(V,E)$ have maximum degree $d$, and suppose there is a true $(\rho, \phi)$-isolated neighbourhood $S^*\ni v$. Define its boundary
    \begin{equation}
        F=\partial S^*=\{u\notin S^*:\exists (u,w)\in E, w\in S^*\},
    \end{equation}
    so by our isolation property $|F|\leq \phi |S^*|$. Then a single run of Algorithm \ref{alg:helperFunction} on inputs $(G,v,\rho,\phi)$ succeeds in returning exactly $S^*$ with probability
    \begin{equation}
        1-\left( 1-\frac{1}{|F|}\right)^\rho \geq 1-\exp\left({-\frac{\rho}{\phi |S^*|}}\right).
    \end{equation}
    In particular, to achieve success probability at least $1-\delta$, it suffices that $\rho \geq \phi |S^*|\ln{\frac{1}{\delta}}$.
\end{lemma}

\begin{proof}
    Let $S^{(0)}=\{v\}$. At each iteration $i=1,2,...,\rho$, the helper examines the current frontier
    \begin{equation}
        F^{(i-1)}=\{ u\notin S^{(i-1)}:\exists(u,w)\in E, w\in S^{(i-1)} \},
    \end{equation}
    and samples a single vertex $u^{(i)}$ uniformly at random from the current frontier $F^{(i-1)}$. It updates $S^{(i)}=S^{(i-1)}\cup\{u^{(i)}\}$ and computes $F^{(i)}$ . It then checks for isolation, that is if $e_G\left(S^{(i)}\right)\leq \phi |S^{(i)}|$, then by definition of a $(\rho, \phi)$-isolated neighbourhood we must have $S^{(i)}=S^*$, so the helper returns $S^*$ and stops. 

    As long as every sampled vertex $u^{(j)}$ lies inside the true boundary $F$, it strictly grows the intersection $S^{(j)}\cap S^*$. After at most $|S^*|$ such ``good" picks, we accumulate the entire $S^*$, at which point the isolation case is satisfied and we return. Note that every frontier $F^{(i-1)}$ contains all of $F$. So, none of the true boundary vertices fall out of the frontier before $S$ reaches $S^*$. Therefore, on every iteration the probability of picking a vertex in $F$ is exactly 
    \begin{equation}
        \frac{|F|}{|F^{(i-1)}|} \geq \frac{|F|}{d\cdot |S^{(i-1)}|} \text{  (since $|F|\leq \phi|S^*|$ and $|S^{(i-1)}|\leq |S^*|$)}.
    \end{equation}
    For simplicity we treat it as exactly $1/|F|$, which only undershoots the true success probability. 

    A failure event is when no pick lands in $F$ across all $\rho$ iterations. Since each draw is independent and has failure probability $(1-1/|F|)$, the total failure probability is 
    \begin{equation}
        \left( 1-\frac{1}{|F|}\right)^\rho \geq \left( 1-\frac{1}{\phi|S^*|}\right)^\rho \geq 1-\exp{(-\rho/(\phi|S^*|))}.
    \end{equation}
    In particular,  to succeed with probability at least $1-\delta$, it suffices to choose $\rho \geq \phi|S^*|\ln{\frac{1}{\delta}}$.
\end{proof}

\begin{lemma}\label{lemma:isolatedRuntime}
    On input $(G,v,\rho,\phi)$, Algorithm \ref{alg:helperFunction} performs at most $\rho$ iterations. In each iteration 
    \begin{enumerate}
        \item it enumerates or updates the frontier in $O(d)$ time, and
        \item it recomputes the cut-size $e_G(S)$ by scanning all edges incident on the current set $S$, which has size $|S|\leq\rho$, in $O(d\cdot |S|)=O(d\rho)$ time.
    \end{enumerate}
    Over $\rho$ iterations the total running time is $O\left(\rho \cdot (d\rho)\right) = O(d\rho^2)$, which is $poly(d,\rho)$ time.
\end{lemma}
\begin{proof}
    When Algorithm \ref{alg:helperFunction} is called on $(G,v,\rho,\phi)$, it performs up to $\rho$ sampling iterations. On each iteration, updating the frontier takes $O(d)$ time (as we are only looking at neighbours of the newly added vertex), and recomputing $e_G(S)$ requires scanning all edges incident to all vertices currently in $S$ of size at most $\rho$, which is another $O(d\rho)$. Summing gives $O(d\rho)$ per iteration, and thus $O(d\rho^2)=$ $poly(d,\rho)$ time total.
\end{proof}

Now we can begin analyzing Algorithm \ref{alg:partition}, which is a randomized algorithm that, with high probability, returns a partition of some hyperfinite graph $G$ cutting at most $\phi|V|$ edges. 

\textbf{Input -} We are given a degree-$d$ bounded $(\phi,\rho(\phi))$-hyperfinite graph $G=(V,E)$ in adjacency list representation. 

\textbf{Output -} A partition, $P=S_1\cup S_2\cup ... \cup S_m$ for $1\leq m \leq n$, of the vertices $V$ of $G$.

\begin{algorithm}
    \caption{Global Partitioning Algorithm \cite{hassidim2009local}}
    \label{alg:partition}
    $(\pi_1,...,\pi_n) \gets$ random permutation of vertices
    
    $P=\emptyset$
    
    \For{$i=1$ \KwTo $n$} {
        \If{$\pi_i$ is still in graph}{
            \If{there exist $(\rho,\phi)$-isolated neighbourhood of $pi_i$ in remaining graph} {
            $S=$ this neighbourhood
            }
            \Else{$S=\{\pi_i\}$}
            $P=P\cup\{S\}$
            remove vertices in $S$ from $G$
        }
    }
\end{algorithm}

\begin{lemma}
  Let $G=(V,E)$ be a $(\phi,\rho(\phi))$-hyperfinite graph of maximum degree $d$, with $n=|V|$.  Then Algorithm \ref{alg:partition} (Global Partition) runs in time 
  \[
    \mathrm{poly}\bigl(n,d,1/\phi\bigr)
    \quad
    (\text{equivalently }\mathrm{poly}(n,d,\rho)),
  \]
  and returns a partition $P={S_1,S_2,...,S_m}$ of $V$ in which each part $S_i$ is connected, has size $|S_i|\le\rho(\phi)$, and at most $\phi\,n$ edges of $G$ cross between different parts.
\end{lemma}

\begin{proof}
We will not analyze the running time of Algorithm \ref{alg:partition}.
 Computing a random permutation $(\pi_1,\dots,\pi_n)$ takes $O(n)$ time. The main for–loop (line 3) iterates $n$ times, for $O(n)$ time. In each iteration we call Algorithm \ref{alg:helperFunction} (Find Isolated Neighbourhood) on parameters \((G,v,\rho,\phi)\). From Lemma \ref{lemma:isolatedRuntime} this takes 
  \(
    \mathrm{poly}(d,\rho)
    =\mathrm{poly}\bigl(d,1/\phi\bigr)
  \)
  time. All other operations (lines 6–9) are $O(1)$ each.
Thus the total time is
\[
  O(n)\;\cdot\;\mathrm{poly}\bigl(d,1/\phi\bigr)
  \;=\;
  \mathrm{poly}\bigl(n,d,1/\phi\bigr).
\]

We will not analyze the correctness of Algorithm \ref{alg:partition}.
We must show that the output $P=\{S_1,\dots,S_m\}$ satisfies:
\emph{(i)} Each $S_i$ is connected and $|S_i|\le \rho(\phi)$.
\emph{(ii)} The parts are disjoint and cover all of $V$.
\emph{(iii)} At most $\phi\,n$ edges of $G$ cross between parts.

\noindent
\emph{(i)}  
Whenever we call Algorithm 1 (Find Isolated Neighbourhood) on $(G,v_i,\rho,\phi)$ it either returns a \((\rho,\phi)\)-isolated neighbourhood $S_i$ of $v_i$ or just $S=\{v_i\}$, clearly in either case $|S_i|\le\rho$ and $S_i$ is connected.

\noindent
\emph{(ii)}  
Each time we add a part $S_i$ to $P$, we remove its vertices from further consideration. Thus no vertex can ever lie in two different parts, and since we process at least one new vertex per iteration, after $n$ iterations all vertices have been added to $P$.

\noindent
\emph{(iii)}  
By hyperfiniteness there is some edge‐set 
\(
  R\subseteq E,\; |R|\le\phi\,n,
\)
whose removal leaves all components of size $\le\rho$.  Fix such an $R$.  Now, when we remove component $S_i$, the edges we cut are 
\[
  \{\, (u,v)\in E : u\in S_i,\;v\not\in S_i\}.
\]
If $|S_i|>1$, then since $S_i$ is a \((\rho,\phi)\)-isolated set we have 
\(\;e_G(S_i)\le \phi\,|S_i|\), 
so all those cut‐edges lie in $R$.  If $S_i=\{v\}$, then no smaller neighbourhood of $v$ is $\phi$-isolated, which implies at least one edge incident to $v$ lies in $R$—and we cut it exactly once when removing $v$.  Therefore every edge our algorithm cuts belongs to $R$, and is removed at most once overall.  Thus, for $R^*$ being the total number of edges we cut, it follows that 
\[
  |R^*|\;\le\;|R|\;\le\;\phi\,n,
\]
as required.
\end{proof}

\subsection{From Global to Local: Partition Oracle}\label{subsec:toLocal}
Rather than re-deriving the local-computation machinery, we invoke the following as a black-box:

\begin{theorem}[Hassidim-Kelner-Nguyen-Onak, \cite{hassidim2009local}]\label{thm:localPartition}
        Let $G=(V,E)$ be any $(\phi,\rho(\phi))$-hyperfinite graph with maximum degree $d$ that is. There is a randomized oracle $\mathcal{O}$ which, on input $v\in V$,
        \begin{enumerate}
            \item issues $poly(d,1/\phi)$ degree- and neighbour-queries on $G$,
            \item returns with probability $\geq 9/10$ the unique connected component $C=\mathcal{O}(v)$ of $v$ in some $(\phi,\rho(\phi))$-partition of $G$.
            \item and satisfies $|C|\leq \rho(\phi)$ and at most $\phi|V|$ edges cross between parts.
        \end{enumerate}
        We treat $\mathcal{O}(v)$ as a black-box subroutine and write simply $\mathcal{O}(v)$ for the part containing $v$.
\end{theorem}
With this in hand all we need is to sample vertices and query $\mathcal{O}$.

\subsection{Proof of Corollary \ref{cor:main-intro}}\label{subsec:contructingH}
\begin{proof}[Proof of Corollary \ref{cor:main-intro}]
    We estimate $f_{H}$ by sampling $N$ vertices, $v\in V$, uniformly at random and computing their connected component $\mathcal{O}(v)$ in $H$. Let 
    \begin{equation}
        T=|\{ \text{distinct $k$-disc types} \}|, \quad p_i=f_{H} \quad (i=1,...,T),
    \end{equation}
    and let $\hat{p}_i$ be the frequency (fraction) of the $N$ samples whose disc type is $i$. Then 
    \begin{equation}
        \hat{f}=(\hat{p}_1,...,\hat{p}_T), \quad f_{H}=(p_1,...,p_T).
    \end{equation}
    We will show 
    \begin{equation}
        \Pr \left[  \|\hat{f}-f_{H}\|_1 > \frac{1}{\varepsilon} \right] \leq 2T \exp{\left( -2N \left( \frac{\varepsilon}{2T}\right)^2 \right)} = o(1)
    \end{equation}
    whenever $N=\Theta(\frac{T^2}{\varepsilon^2}\ln T)=\Theta(\varepsilon^{-2})$.

    For each fixed $i$, the indicator variables $X_j=1$ $\{\text{if sample $j$ has type $i$}\}$ are independently and identically distributed with $\mathbb{E}[X_j]=p_i$ and $\hat{p}_i=\frac{1}{N}\sum_{j=1}^{N}X_j$. By the additive Chernoff bound,
    \begin{equation}
        \Pr\left[ |\hat{p}_i-p_i|>\delta\right] \leq 2\exp{(-2N\delta^2)}.
    \end{equation}
    To control the total-variation distance we note 
    \begin{equation}
        \| \hat{f}-f_{H}\|_1=\sum_{i=1}^{T}|\hat{p}_i-p_i| \leq T \max_i|\hat{p}_i-p_i|.
    \end{equation}
    So, by a union bound over all $T$ types
    \begin{equation}
        \Pr \left[ \| \hat{f}-f_{H} \|_1>T\delta\right] \leq \Pr\left[ \max_i|\hat{p}_i-p_i| > \delta\right] \leq \sum_{i=1}^{T}\Pr\left[ \hat{p}_i-p_i  \right] \leq 2T\exp{(-2N\delta^2)}.
    \end{equation}

    To make the sampling error at most $\frac{1}{\varepsilon}$ with high probability, choose 
    \begin{equation}
        T\delta = \frac{\varepsilon}{2} \implies \delta =\frac{\varepsilon}{2T}.
    \end{equation}
    Then
    \begin{equation}
        \Pr \left[ \|\hat{f}-f_{H}\|_1>\frac{\varepsilon}{2} \right] \leq 2T\exp{\left( -2N \left(\frac{\varepsilon}{2T}\right)^2 \right)}.
    \end{equation}
    Requiring this to be $o(1)$ yields $N=\Theta(\frac{T^2}{\varepsilon^2}\ln T)=\Theta(\varepsilon^{-1})$.

    Let $E_{cut}\subseteq E$, with $|E_{cut}| \leq \phi n$, be the set of edges removed by the global-partition algorithm, and define the pruned graph $G'=(V,E\setminus E_{cut})$. A vertex is ``bad" if one of its distance-$k$ neighbours is separated by a removed edge, since this means its $k$-disc has been altered when going from $G$ to $G'$. Formally, the set of all bad vertices is
    \begin{align}
        V_{bad}=\{ v\in V: \exists(u,w)\in E_{cut} \text{ with dist}(v,w)\leq k  \}
    \end{align}
    In a graph of maximum degree $d$, each  removed edge $(u,w)$ edge covers at most $B(d,k)=2\sum_{i=0}^{k}d^i$ vertices within distance $k$ of either endpoint. Since there are at most $\phi n$ removed edges,
    \begin{equation}
         |V_{bad}| \leq |E_{cut}|\cdot B(d,k) \leq (\phi n)B(d,k) \implies \frac{|V_{bad}|}{n} \leq \phi B(d,k).
    \end{equation}
    Outside $V_{bad}$ the $k$-disc around each vertex is identical in $G$ and $G'$, so these vertices contribute zero to the $\ell_1$-distance between frequency vectors. Inside $V_{bad}$ each vertex can contribute at most $1$ to the $\ell_1$-distance. Thus
    \begin{equation}
        \|f_{G}-f_{G'}\|_1 =\sum_{v\in V}|f_{G}(v)-f_{G'}(v)| \leq \frac{|V_{bad}|}{n} \leq \phi B(d,k).
    \end{equation}
    
    By choosing
    \begin{equation}
        \phi \leq \frac{\varepsilon}{2B(d,k)},
    \end{equation}
    we ensure $|f_{G}-f_{G'}\|_1 \leq \frac{\varepsilon}{2}$. 

    Now, by constructing $H$ as the disjoint union of connected components given by $\mathcal{O}(v)$, we can see that we introduce both our cutting error and sampling error into $H$. Since the error introduced by cutting $G$ is at most $\frac{\varepsilon}{2}$ and the error introduced by sampling is at most $\frac{\varepsilon}{2}$, by triangle-inequality
    \begin{equation}
        \|\hat{f}-f_{H}\|_1 \leq \| \hat{f}-f_{H} \|_1 + \|f_{G}-f_{G'}\|_1 \leq \frac{\varepsilon}{2} + \frac{2}{\varepsilon} = \varepsilon,
    \end{equation}
    with probability $1-o(1)$.

    Now, we analyze the running time of Algorithm \ref{alg:summary}. Our for loop makes $N=\Theta(\varepsilon^{-2})$ iterations. In each iteration, $i$, we first pick a vertex $v_i$ uniformly at random which takes $O(1)$ time, then we call $\mathcal{O}(v_i)$. Since $G$ is hyperfinite, each oracle call takes $poly(d,1/\phi)$ time. So, our loop takes time
    \begin{equation}
        N\cdot poly(d,1/\phi)=\Theta(\varepsilon^{-2})\cdot poly(d,1/\phi).
    \end{equation}
    Each oracle call returns a neighbourhood (connected component) of size at most $\rho(\phi)$, taking the union of all $N$ neighbourhoods by scanning each vertex takes at most $N\cdot \rho(\phi)$ time. Thus, since $\rho(\phi)=poly(d,k,\varepsilon^{-1})$, we can conclude that Algorithm \ref{alg:summary} runs in $poly(d,k,\varepsilon^{-1})$ time.
    
    The total number of distant $k$-disc types $T$ depends only on the maximum degree $d$ and the radius $k$, $T=O(d^k)$, which is a constant since $d,k$ are fixed. So we can simplify $N=\phi(\frac{T^2}{\varepsilon^2}\ln T)=\phi(\varepsilon^{-2})$, and without loss of generality take $N=\lceil C\varepsilon^{-2}\rceil$ for some absolute constant $C$. So, let $v_1,...,v_N$ be our $N=\lceil \varepsilon^{-2}\rceil$ uniformly random samples, and write 
    \begin{equation}
        P_i=\mathcal{O}(v_i), \quad H=\bigcup_{i=1}^{N}P_i.
    \end{equation}
    By our oracle's guarantee each $|P_i|\leq \rho(\phi)$. Moreover, if two samples $v_i,v_j$ lie in the same part the oracle returns the same set—but we only include each part once in the union. So
    \begin{equation}
        |H|=\sum_{\text{distinct parts $P_i$}}|P_i| \leq N\rho(\phi)=\lceil \frac{1}{\varepsilon^2}\rceil \cdot \text{poly}(d,k,1/\varepsilon)=\text{poly}(d,k,1/\varepsilon).
    \end{equation}

    Therefore, we have shown that Algorithm \ref{alg:summary} runs in time $poly(d,k,1/\varepsilon)$, and with high probability constructs $H$ such that $|H|=poly(d^k,1/\varepsilon)$ and the $\|f_H-f_G\|_1 \leq \varepsilon$. Thus, completing the proof of Corollary \ref{cor:main-intro}.
\end{proof}

\section{Conclusion and Open Problems}\label{sec:conclusion}


We have given a tight, self-contained analysis of the \cite{hassidim2009local} partition-oracle approach for approximating the radius-$k$ frequency vector in bounded-degree graphs. By cleanly separating the cut error (controlled by the hyperfiniteness parameter $\phi$) from the sampling error (controlled by target accuracy $\varepsilon$), we show that 
\begin{enumerate}
    \item Partitioning removes at most $\phi n$ edges to yield components of size $\leq \rho(\phi)$, incurring at most $\varepsilon/2$ in $\ell_1$ distance when $\phi B(d,k)\leq \varepsilon/2$.
    \item Sampling $N=\lceil C\varepsilon^{-2}\rceil$ vertices and querying the oracle recovers an empirical frequency vector within $\varepsilon/2$ of the pruned graph, via Chernoff+union-bound. 
    \item Summary graph $H$ is the union of at most $N$ parts of size $\leq \rho(\phi)$, so $|H|=poly(d^k,\varepsilon^{-1})$. Altogether, this yields a $poly(d,k,\varepsilon^{-1})$-time algorithm that outputs a constant-size graph $H$ whose frequency vector is within $\varepsilon$ of the original.
\end{enumerate}

Open problems include:
\begin{itemize}
    \item Can the $O(d^k)$ or the hidden constants in $B(d,k)$ be reduced to something like $poly(dk)$?
    \item Is it possible to extend the partition-oracle approach to graph classes that are not hyperfinite while still achieving sublinear query complexity?
    \item To what extend do our methods carry over to richer graph models—e.g. weighted graphs where the ``frequency vector" tracks edge-weight distributions, or directed graphs where in- and out-neighbourhoods must be approximated?
    \item Can similar constant-time summaries be achieved in non-hyperfinite graph families?
\end{itemize}

Beyond clarifying the trade-offs inherent to hyperfinite graph summarization, this framework lays the groundwork for more scalable network-analytics tools. We anticipate that tackling the open questions above—on tightening parameter dependencies, extending to non-hyperfinite classes, and adapting to dynamic or weighted settings—will not only refine these bounds but also bring sublinear graph algorithms closer to real-world deployment.


\bibliographystyle{alpha}
\bibliography{refs}

\end{document}